\pdfoutput=1

\documentclass[11pt,a4paper]{article}

\usepackage[T1]{fontenc}
\usepackage[utf8]{inputenc}
\usepackage[margin=2.9cm]{geometry}
\usepackage{amsmath,amssymb,amsthm}
\usepackage{mathtools}
\usepackage{booktabs}
\usepackage{enumitem}
\usepackage[colorlinks=true,linkcolor=blue,citecolor=blue,urlcolor=blue]{hyperref}

\theoremstyle{plain}
\newtheorem{theorem}{Theorem}[section]
\newtheorem{proposition}[theorem]{Proposition}
\newtheorem{lemma}[theorem]{Lemma}
\newtheorem{corollary}[theorem]{Corollary}

\theoremstyle{definition}
\newtheorem{definition}[theorem]{Definition}
\newtheorem{problem}[theorem]{Problem}

\theoremstyle{remark}
\newtheorem{remark}[theorem]{Remark}

\DeclareMathOperator*{\argmax}{arg\,max}
\DeclareMathOperator{\diam}{diam}

\newcommand{\R}{\mathbb{R}}
\newcommand{\Q}{\mathbb{Q}}

\newcommand{\N}{\mathbb{N}}
\newcommand{\ER}{\exists\R}
\newcommand{\FER}{\forall\R}
\newcommand{\EFR}{\exists\forall\R}
\newcommand{\NNReach}{\textsc{NNReach}}
\newcommand{\VIP}{\textsc{VIP}}
\newcommand{\NE}{\textsc{NE}}
\newcommand{\MIN}{\textsc{Min}}
\newcommand{\NI}{\textsc{NonInt}}
\newcommand{\MONO}{\textsc{Mono}}
\newcommand{\CF}{\textsc{CFair}}
\newcommand{\TRIG}{\textsc{Trigger}}
\newcommand{\UVUL}{\textsc{UnivVuln}}
\newcommand{\INV}{\textsc{InvRes}}
\newcommand{\IND}{\textsc{OutInd}}
\newcommand{\WM}{\textsc{WMRemove}}
\newcommand{\REP}{\textsc{Repair}}
\newcommand{\WRE}{\textsc{WFault}}
\newcommand{\WRA}{\textsc{WSafe}}
\newcommand{\ReLU}{\mathrm{ReLU}}
\newcommand{\ETR}{\textsc{Etr}}
\newcommand{\ETRINV}{\textsc{Etr-Inv}}
\newcommand{\RETRINV}{\textsc{Range-Etr-Inv}}
\newcommand{\Sig}[1]{\Sigma_{#1}^{\mathrm P}}
\newcommand{\Pih}[1]{\Pi_{#1}^{\mathrm P}}

\title{Security Properties of Neural Networks\\ as Decision Problems}
\author{Adrian Wurm\\
\small BTU Cottbus--Senftenberg, Lehrstuhl Theoretische Informatik\\
\small Platz der Deutschen Einheit 1, 03046 Cottbus, Germany\\
\small wurm@b-tu.de}
\date{25 September 2026}

\begin{document}
\maketitle

\begin{abstract}
Certifying a deployed neural network raises decision problems that the
verification literature has not classified: whether the model carries a
backdoor planted in its training data, whether a fault in
its stored parameters can drive it into an unsafe state, whether its output
leaks a private part of its input. We formalise eight such
problems and classify what we can.

The organising observation is a logical one. The function computed by a
piecewise linear network, together with all its node values, is definable
by a quantifier-free formula of real addition of size linear in the
network, so a property of the network is a quantifier-alternation sentence,
which Sontag's 1985 theorem places in the polynomial hierarchy at the level
of its prefix. Membership results are thus corollaries, and the
argument makes plain what they need: that the quantified objects are
inputs rather than the network's own parameters.

Non-interference, monotonicity and counterfactual fairness have exactly the
complexity of network equivalence and of interval verification, all co-NP-complete over
$\ReLU$. Detection of backdoor triggers from a
quantised alphabet is $\Sig2$-complete, one level above robustness
certification, so it does not reduce to polynomially many robustness
queries unless the hierarchy collapses. Inversion resistance is
co-NP-complete for every $\ell_p$ metric, $p$ a fixed positive integer.
Quantifying over \emph{parameters} instead of inputs - the fault model of
bit-flip attacks, radiation upsets and analog accelerators - makes
verification $\ER$-complete already for networks of identity nodes, for
which every previously studied problem is in $\mathrm P$, and it stays so
when each parameter is confined to a box of inverse-polynomial width; the
corresponding safety question is $\FER$-complete for $\ReLU$.
\end{abstract}

\begingroup
\renewcommand{\thefootnote}{}
\footnotetext{\emph{ACM CCS:} Theory of computation $\to$ Complexity
classes; Logic; Machine learning theory.}
\endgroup

\section{Introduction}

Let $N$ be a feedforward neural network whose activations are piecewise
linear with rational coefficients. The graph of $N$ - the set of pairs
consisting of an input and the induced vector of node values - is then
definable by a quantifier-free formula in the language of real addition,
of size linear in $N$: each node contributes a Boolean combination of
linear equations and inequalities of constant size, and the network
contributes their conjunction. A property of $N$ of the kind verification
asks about is therefore a sentence of real addition whose quantifier
prefix is the prefix of the property and whose quantifier-free part names
the network.
By Sontag's theorem \cite{Sontag}, deciding $k$-alternation sentences of
real addition is complete for $\Sig k$, so such a property lies at the
level of the polynomial hierarchy that its prefix dictates, and nowhere
else.

This is the organising fact of the present paper, and it has two
consequences that pull in opposite directions. On the one hand it makes
membership results automatic: the certificate arguments that the
verification literature has been reproving for a decade - guess the
activation pattern and the violated constraint, then solve a linear
program - are the case $k=1$ of a theorem from 1985, and the higher levels
come for free (Lemma~\ref{lem:uniform}). On the other hand it delimits
sharply where the polynomial hierarchy is the right home for such a
problem at all. The defining formula is a formula of real \emph{addition}
only because the quantified objects are inputs, which the weights
multiply by constants. Quantify instead over the network's own parameters,
and each node contributes a product of two unknowns; the defining formula
becomes polynomial, the problem leaves real addition, and its natural home
is the existential theory of the reals.

The paper develops both consequences by way of a schema. All the problems
we consider, together with all the previously studied ones, are instances
of
\[
  Q_1\,p\in\mathcal P\quad Q_2\,x\in A\ :\qquad \Phi\bigl(N_p,x\bigr),
\]
in which $x$ is an input, $p$ is a parameter that perturbs the situation,
$Q_1,Q_2$ are quantifiers, and $\Phi$ is a condition expressible by linear
constraints on the values of a network at a point. Two independent axes
govern the complexity. The first is the quantifier prefix, which places
the problem in the polynomial hierarchy. The second, which has not been
isolated before, is whether $p$ enters the computation linearly or
multiplicatively: input-space parameters enter linearly, parameter-space
parameters multiply activations, and that distinction alone decides
whether the problem lives in the polynomial hierarchy or in the real
hierarchy. Theorem~\ref{thm:wr} shows that the second axis is not a
technicality. For networks all of whose nodes compute the identity, every
problem previously studied lies in $\mathrm P$; making the parameters
uncertain, within boxes of inverse-polynomial width, makes verification
$\ER$-complete. Nothing about the activation function is involved, which
is exactly why the phenomenon is invisible in a classification indexed by
activation functions.

What the schema needs in order to be more than a bookkeeping device is a
supply of instances that actually occupy its cells, and the previously
studied problems do not: reachability, verification of an interval
property, network equivalence and the robustness notions all have a fixed
network and a single quantifier, so they sit in one corner of the grid,
with network minimisation the lone exception at the second level. Problems
with genuinely alternating prefixes, and problems that quantify over the
parameters, are supplied instead by the questions that certification of a
deployed model turns on. Whether a network carries a backdoor planted by
whoever supplied its training data is an $\exists\forall$ question over a
trigger. Whether a fault injected into the memory holding the parameters
can drive the network into an unsafe state is an $\exists\exists$ question
over the parameters. Whether a model reveals the part of its input that
was to stay private, whether a stolen copy can be stripped of its
watermark, and whether a network that failed certification can be patched
without invalidating the rest of the safety case, are further instances,
and they fall in four different cells. This paper proposes formalisations
of eight such properties, classifies what it can, and records in
Table~\ref{tab:schema} where each one lands. Three of them turn out to be
old problems in new clothes, which is worth knowing precisely because it
tells the reader which of the eight are new.

\subsection{Contributions}

\begin{enumerate}[label=(C\arabic*),leftmargin=*]
\item A uniform membership lemma (Lemma~\ref{lem:uniform}): for networks
whose activations are semilinear with rational coefficients, the graph of
the network is definable by a quantifier-free formula of real addition of
size linear in the network. Every membership result for such networks in
the polynomial hierarchy is therefore a corollary of Sontag's theorem
\cite{Sontag} that $k$-alternation sentences of real addition are complete
for $\Sig k$. The certificate arguments in the verification literature are
its $k=1$ case.

\item Three security properties turn out to have exactly the complexity of
problems that are already classified
(Section~\ref{sec:notnew}): non-interference that of network equivalence
(Proposition~\ref{prop:NI}), and monotonicity and counterfactual fairness
that of verification of an interval property
(Proposition~\ref{prop:mono}); all three are co-NP-complete over $\ReLU$
networks, and the hardness proofs are given for $\ReLU$ alone, with no
threshold activation assumed. The
equivalence is one of complexity only. The three are different
specifications, stated over different data and naturally implemented by
different front-ends; what the reductions give is that the classification
by activation function transfers to them verbatim, and that an existing
equivalence or interval checker decides them without being modified. For
monotonicity over $\ReLU$ networks the co-NP-completeness is already known
from the study of extracted rules \cite{WurmRules}, with which
Remark~\ref{rem:rules} compares.

\item Backdoor-trigger detection with a quantised trigger alphabet is
$\Sig 2$-complete for $\ReLU$ networks (Theorem~\ref{thm:trig}). We give
the reduction in full, and show in Section~\ref{sec:relax} that the
quantisation is not a technical convenience: the continuous relaxation
admits a strategy for the adversary - setting a literal and its negation
both to zero - that has no Boolean counterpart, so the relaxed problem
is a different one.

\item Inversion resistance is co-NP-complete
(Theorem~\ref{thm:inv}).

\item If the parameters carry the uncertainty, verification is
$\ER$-complete already for networks of identity nodes
(Theorem~\ref{thm:wr}), for which every previously studied verification
problem lies in $\mathrm P$, and it remains so when every parameter is
confined to a box of inverse-polynomial width. The corresponding safety
question is $\FER$-complete for $\ReLU$. The proof is a direct reduction
from the range-restricted inversion problem $\RETRINV$ \cite{AMS}.

\item Definitions and initial classifications for universal vulnerability,
output indistinguishability, watermark removability and repair, the last
of which lands naturally in the second level of the \emph{real} hierarchy
(Section~\ref{sec:repair}).
\end{enumerate}

\subsection{Related work}

On the logical side the paper rests on two classical results and their
surrounding literature. Sontag's theorem on real addition \cite{Sontag} is
what places the first-order properties of a piecewise linear network in the
polynomial hierarchy, and Lemma~\ref{lem:uniform} is little more than the
observation that the hypotheses of that theorem are met uniformly. For
$\ER$ as a complexity class, and for the catalogue of problems complete for
it, we follow the compendium of Schaefer, Cardinal and Miltzow
\cite{Compendium}; for the hierarchy above $\ER$ we follow Schaefer and
\v{S}tefankovi\v{c} \cite{BeyondER}. The complete problem we reduce from in
Section~\ref{sec:weights} is the range-restricted inversion problem of
Abrahamsen, Miltzow and Seiferth \cite{AMS}, in the streamlined form of
\cite{AMtoolbox}, whose range promise is what lets us shrink the parameter
boxes.

On the verification side, the complexity of reachability for $\ReLU$
networks is due to Katz et al.\
\cite{Katz} and S\"alzer and Lange \cite{SL}; the classification across
activation functions, and the problems $\VIP$, $\NE$, $\MIN$ and the
robustness family, are from \cite{Wurm23,Wurm24,WurmJ}, whose notation we
follow throughout. Training $\ReLU$ networks is $\ER$-complete
\cite{AKM}; Theorem~\ref{thm:wr} isolates the reason, which is not the
optimisation but the fact that parameters and activations multiply.
The complexity of checking rules extracted from a network is studied in
\cite{WurmRules}. The algorithmic side has matured in parallel: complete
verifiers built on SMT solving \cite{Katz,Marabou}, on mixed integer
programming \cite{Tjeng} and on branch and bound over bound propagation
\cite{BetaCrown} are compared annually in the international verification
competition \cite{VNNCOMP}. The lower bounds below say what those tools
cannot be extended to do without leaving the arithmetic they are built on.

Each property we consider has a literature of its own, which is
algorithmic or empirical rather than complexity-theoretic. Non-interference
originates with Goguen and Meseguer \cite{Goguen} and is surveyed for
programming languages by Sabelfeld and Myers \cite{Sabelfeld}; its
quantitative form, information leakage, has for neural networks been
approached by approximate model counting \cite{Baluta}. Monotonicity is enforced by
construction or certified post hoc in \cite{CertMono,CEGMono};
counterfactual fairness is due to Kusner et al.\ \cite{CFair}, individual
fairness to Dwork et al.\ \cite{FTA}, and its verification is treated in
\cite{VerifFair}. Backdoor attacks were introduced in \cite{BadNets} and
are scanned for by methods such as Neural
Cleanse \cite{NeuralCleanse}; universal adversarial perturbations are due
to Moosavi-Dezfooli et al.\ \cite{UAP}. Weight faults are studied as an
attack \cite{BitFlip}, as a reliability problem
\cite{ErrorProp}, as a consequence of analog hardware \cite{InMemory}, and
empirically as a robustness notion \cite{WeightPerturb}. Model inversion
\cite{ModelInversion} and membership inference \cite{MembershipInference}
motivate Section~\ref{sec:privacy}, backdoor-based watermarking
\cite{AdiWatermark} motivates Problem~\ref{prob:wm}, and network repair \cite{ProvableRepair,MinimalMod}
motivates Section~\ref{sec:repair}. To the best of our knowledge none of
these problems has previously been placed in a complexity class, with the
qualifications recorded below for backdoors and for parameter
perturbation.

Three lines of work stand close enough to ours that the relationship needs
stating. First, backdoors have been treated formally from two sides. Pham
and Sun \cite{PhamSun} give a verification procedure for what is
essentially our Problem~\ref{prob:trig}, combining statistical sampling
with abstract interpretation to certify that no input-agnostic trigger of a
given shape achieves a given success rate; Theorem~\ref{thm:trig} supplies
the complexity-theoretic floor under that procedure and, through the
argument in Section~\ref{sec:relax}, explains why its outer search over
candidate triggers cannot be replaced by a bounded number of inner
verification calls. Goldwasser, Kim, Vaikuntanathan and Zamir \cite{GKVZ}
prove a sharper-looking statement of an entirely different kind: a backdoor
can be planted so that the resulting network is \emph{computationally
indistinguishable} from an honestly trained one, so that no efficient
detector succeeds at all under standard cryptographic assumptions. Neither
result implies the other. Theirs is an average-case, cryptographic statement about a
particular planted construction, and it rules out efficient detection
outright; ours is an unconditional worst-case statement about the decision
problem itself, and it locates that problem at a specific level of the
polynomial hierarchy, which is what licenses the separation from robustness
certification drawn in Section~\ref{sec:relax}. A $\Sig2$-complete problem
may still have easy instances, and a cryptographically undetectable
backdoor says nothing about where the worst case sits.

Second, the arithmetic in which a network is evaluated is a third axis,
orthogonal to the two studied here. Alsmann, Lange and S\"alzer \cite{ALS}
classify verification when the network is evaluated in finite-width fixed-
or floating-point arithmetic and obtain PSPACE-completeness for bit-vector
specifications. We touch that axis twice, in Section~\ref{sec:relax} and in
Section~\ref{sec:repair}, at the points where quantising a continuous
search space changes a problem's class rather than approximating it.

Third, and closest to Section~\ref{sec:weights}, Soltanalian
\cite{Soltanalian} studies exact $\ReLU$ verification in a smoothed model
in which every weight and bias is independently perturbed by clipped,
rounded Gaussian noise, and shows that no sound and complete verifier runs
in expected polynomial time under the perturbation unless
$\mathrm{NP}\subseteq\mathrm{BPP}$. The question there is different from
ours in a way worth making explicit: the noise is a \emph{preprocessing} of
the instance, applied once and at random, after which the problem asked is
still the input-space question $\forall x\in X$, and the conclusion is that
worst-case hardness survives randomisation of the parameters. Here the
parameters are instead \emph{quantified inside the decision problem}, over
an adversarially chosen uncertainty set, and the conclusion is that the
problem leaves the polynomial hierarchy altogether. The two results point
the same way from opposite directions - randomising the parameters does
not make verification easy, and quantifying over them makes it strictly
harder - and neither subsumes the other.

\section{Preliminaries}

We follow \cite{WurmJ}. A \emph{feedforward neural network} $N$ is a
layered graph computing a function $\R^n\to\R^m$; node $i$ of layer $\ell$
computes $y_{\ell i}=\sigma_{\ell i}(\sum_j c^{(\ell-1)}_{ji}
y_{(\ell-1)j}+b_{\ell i})$ with rational weights $c$ and biases $b$ and
activation functions $\sigma_{\ell i}$ drawn from a set $F$; such a network
is an \emph{$F$-network}. An \emph{LP specification} is a system
$Ax\le b$ of linear inequalities with rational data. We write
$\NNReach(F)$, $\VIP(F)$, $\NE(F)$ and $\MIN(F)$ for reachability,
verification of an interval property, network equivalence and network
minimisation, as in \cite{WurmJ}. The \emph{classification} of an input is
$\argmax_j N(x)_j$, with ties broken in favour of the smallest index, so
that the classification is a total function.

Every node of an $F$-network applies an activation, the output nodes
included. For $F=\{\ReLU\}$ this means that all node values, and in
particular all outputs, are non-negative; linear combinations with negative
coefficients are formed in the linear part of the \emph{next} node, which
is where we use them below.

A set is \emph{semilinear} if it is a Boolean combination of half-spaces;
an activation is semilinear if its graph is. $\ReLU$, Heaviside,
$\mathrm{sign}$, the identity, leaky $\ReLU$ and every rational step
function are semilinear with rational coefficients. We write $\ER$ for the
complexity class of problems reducing in polynomial time to the
existential theory of the reals $\ETR$, and $\EFR$ for the analogous class
one alternation up \cite{BeyondER}.

We use the following results on the two classes involved.

\begin{theorem}[Sontag \cite{Sontag}]\label{thm:sontag}
Fix $k\ge1$ and consider sentences
\[
  Q_1\bar z_1\ Q_2\bar z_2\ \cdots\ Q_k\bar z_k\ F(\bar z)
\]
in which the quantifiers $Q_i$ alternate and $F$ is quantifier-free in the
language of real addition with binary-encoded rational coefficients.
Deciding such a sentence is log-complete for $\Sig k$ when $Q_1=\exists$,
and log-complete for $\Pih k$ when $Q_1=\forall$.
\end{theorem}

\begin{theorem}[Abrahamsen, Adamaszek, Miltzow \cite{AAM}]\label{thm:etrinv}
$\ETRINV$ is $\ER$-complete: given variables $\xi_1,\dots,\xi_k$ confined
to $[\tfrac12,2]$ and a list of constraints, each of the form
$\xi_a+\xi_b=\xi_c$ or $\xi_a\cdot\xi_b=1$, decide satisfiability.
\end{theorem}

Section~\ref{sec:weights} needs the range-restricted refinement of
$\ETRINV$, in which each variable is confined not merely to $[\frac12,2]$
but to a prescribed subinterval of it that may be very short. We follow
\cite{AMS}, with the streamlined exposition in \cite{AMtoolbox}; for $\ER$
and its complete problems generally we follow the compendium
\cite{Compendium}.

\begin{definition}[$\RETRINV$]\label{def:retrinv}
An instance consists of variables $\xi_1,\dots,\xi_k$, for each $i$ a
rational interval $I(\xi_i)\subseteq[\tfrac12,2]$, and a list of
constraints of the two forms above. The question is whether there is an
assignment with $\xi_i\in I(\xi_i)$ for every $i$ satisfying every
constraint. The instance has \emph{range parameter} $\delta$ if
$|I(\xi_i)|\le2\delta$ for every $i$.
\end{definition}

\begin{theorem}[Abrahamsen, Miltzow, Seiferth {\cite[Thm.~3]{AMS}}]
\label{thm:retrinv}
$\RETRINV$ is $\ER$-complete, and remains so when restricted to instances
of range parameter $\delta=O(k^{-c})$, for any constant $c>0$ fixed in
advance. Theorem~\ref{thm:etrinv} is the special case in which every
$I(\xi_i)=[\tfrac12,2]$.
\end{theorem}

\subsection{A network we shall build repeatedly}

Several hardness proofs below run through one $\ReLU$-network encoding of a
$3$-CNF formula, a variant of the network of \cite[Thm.~9(ii)]{WurmJ}. The
statement records only what the later arguments use, namely how the network
behaves on Boolean inputs and the sense in which no real input does better
than the Boolean input it rounds to; the construction itself is deferred to
the proof.

Write $\rho:\R^n\to\{0,1\}^n$ for coordinatewise rounding,
\[
  \rho(z)_a:=1 \text{ if } z_a>\tfrac12, \qquad
  \rho(z)_a:=0 \text{ if } z_a\le\tfrac12,
\]
and, for $b\in\{0,1\}^n$, write $\mathrm{sat}_\varphi(b)$ for the number of
clauses of $\varphi$ satisfied by the assignment $b$.

\begin{lemma}[CNF network]\label{lem:cnf}
Let $\varphi$ be a $3$-CNF formula with $n$ variables and $m\ge1$ clauses.
There is a $\ReLU$-network $M_\varphi:\R^n\to\R$ with $n$ input nodes, two
hidden layers and $4n+2m+1$ further nodes, computable from $\varphi$ in
linear time, such that
\begin{enumerate}[label=(\roman*),leftmargin=*]
\item $M_\varphi(b)=\mathrm{sat}_\varphi(b)$ for every $b\in\{0,1\}^n$;
\item $M_\varphi(z)\le\mathrm{sat}_\varphi(\rho(z))$ for every $z\in\R^n$,
so that $0\le M_\varphi\le m$ everywhere;
\item $M_\varphi(\tfrac12,\dots,\tfrac12)=0$.
\end{enumerate}
\end{lemma}

Items (i) and (ii) are what the identification of $\R^n$ with $\{0,1\}^n$
buys. On the Boolean points the network counts satisfied clauses exactly,
and every other point is dominated by the Boolean point it rounds to, so
nothing is gained by leaving the cube: maximising $M_\varphi$ over $\R^n$
and over $\{0,1\}^n$ give the same value. Item (iii) is of a different
kind, and is the one place where a real input does strictly worse than its
rounding would suggest; Section~\ref{sec:relax} turns on it.

\begin{proof}
For each propositional variable we allocate one input node, for each
literal a pair of nodes in the first hidden layer, for each clause a pair
of nodes in the second hidden layer, and for the evaluation a single node
in the output layer, which follows directly after the second hidden layer.
Linear combinations with negative coefficients are formed in the linear
part of the following node, as everywhere in this paper.

\emph{Literals.} For a variable $a$ with input $z_a$ put
\[
  \alpha_a:=2\ReLU\!\left(z_a-\tfrac12\right)-2\ReLU\!\left(z_a-1\right),
  \qquad
  \alpha_{\neg a}
    :=2\ReLU\!\left(\tfrac12-z_a\right)-2\ReLU\!\left(-z_a\right),
\]
two $\ReLU$ nodes each. So $\alpha_a$ is $2(z_a-\frac12)$ clamped to
$[0,1]$ and $\alpha_{\neg a}$ is $2(\frac12-z_a)$ clamped to $[0,1]$. Three
consequences are used below: $\alpha_\lambda\in[0,1]$ for every literal
$\lambda$; $\alpha_a(z)>0$ exactly when $z_a>\frac12$ and
$\alpha_{\neg a}(z)>0$ exactly when $z_a<\frac12$, so at most one of the
two is non-zero and both vanish precisely at $z_a=\frac12$; and for
$b\in\{0,1\}^n$ one has $\alpha_\lambda(b)=1$ if $\lambda$ is true under
$b$ and $\alpha_\lambda(b)=0$ otherwise.

\emph{Clauses.} For a clause $C=(\lambda_1\vee\lambda_2\vee\lambda_3)$ put
$s_C:=\alpha_{\lambda_1}+\alpha_{\lambda_2}+\alpha_{\lambda_3}$ and
$\gamma_C:=\ReLU(s_C)-\ReLU(s_C-1)$, again two $\ReLU$ nodes, so that
$\gamma_C$ is $s_C$ clamped to $[0,1]$.

\emph{Output.} $M_\varphi:=\sum_C\gamma_C$, one node, whose argument is
non-negative so that the $\ReLU$ on it acts as the identity. The node count
is $4n+2m+1$.

(i) For $b\in\{0,1\}^n$ each $\alpha_\lambda(b)$ is $1$ or $0$ according as
$\lambda$ is true or false under $b$. Hence $s_C\ge1$ exactly when $C$ is
satisfied by $b$, so $\gamma_C=1$ for the satisfied clauses and
$\gamma_C=0$ for the others, and the output is $\mathrm{sat}_\varphi(b)$.

(ii) Fix $z$ and let $C$ be a clause false under $\rho(z)$. Every literal
$\lambda$ of $C$ is false under $\rho(z)$: if $\lambda=a$ this means
$z_a\le\frac12$, hence $\alpha_a(z)=0$; if $\lambda=\neg a$ it means
$z_a>\frac12$, hence $\alpha_{\neg a}(z)=0$. So $s_C=0$ and $\gamma_C=0$.
Every remaining clause contributes at most $1$, which gives
$M_\varphi(z)\le\mathrm{sat}_\varphi(\rho(z))$, and the latter is at most
$m$.

(iii) Taking $z_a=\frac12$ for every $a$ makes every $\alpha_\lambda$
vanish, hence every $s_C=0$ and every $\gamma_C=0$.
\end{proof}

The form in which Section~\ref{sec:trig} uses the lemma is the following.

\begin{corollary}\label{cor:cnfmax}
Let $S$ be a set of variables of $\varphi$ and $\sigma\in\{0,1\}^S$. The
supremum of $M_\varphi$ over the inputs whose $S$-coordinates are fixed to
$\sigma$ equals $m$ if the restriction $\varphi|_\sigma$ is satisfiable,
and is at most $m-1$ otherwise. In the first case it is attained at a
Boolean input.
\end{corollary}

\begin{proof}
Rounding fixes Boolean coordinates, so if $z$ has its $S$-coordinates equal
to $\sigma$ then so does $\rho(z)$. If $\varphi|_\sigma$ is satisfiable,
extend $\sigma$ to a satisfying assignment $b$ and apply
Lemma~\ref{lem:cnf}(i) to get $M_\varphi(b)=m$. If it is not, then every
Boolean point agreeing with $\sigma$ leaves some clause unsatisfied, so
$\mathrm{sat}_\varphi(\rho(z))\le m-1$ for every admissible $z$, and
Lemma~\ref{lem:cnf}(ii) bounds $M_\varphi$ by $m-1$ there.
\end{proof}

The gadget the hardness proofs of Section~\ref{sec:notnew} actually use is
the following non-negative indicator, which vanishes identically exactly
when $\varphi$ is unsatisfiable.

\begin{definition}[Violation gadget]\label{def:gadget}
For a $3$-CNF formula $\varphi$ with $m$ clauses put
$\hat g_\varphi:=\ReLU\bigl(M_\varphi-(m-1)\bigr)$.
\end{definition}

\begin{corollary}\label{cor:gadget}
$\hat g_\varphi$ is a $\ReLU$-network with values in $[0,1]$, computable
from $\varphi$ in linear time, with $\hat g_\varphi\equiv0$ if $\varphi$ is
unsatisfiable, and attaining both the value $1$ and the value $0$ if
$\varphi$ is satisfiable.
\end{corollary}

\begin{proof}
Values in $[0,1]$ because $0\le M_\varphi\le m$ by
Lemma~\ref{lem:cnf}(ii). If $\varphi$ is unsatisfiable then
$\mathrm{sat}_\varphi(\rho(z))\le m-1$ for every $z$, so $M_\varphi\le m-1$
everywhere by (ii) and the $\ReLU$ clips to $0$. If $\varphi$ is
satisfiable, (i) applied to a satisfying assignment gives an input with
$M_\varphi=m$ and hence $\hat g_\varphi=1$, and (iii) gives an input with
$M_\varphi=0$ and hence $\hat g_\varphi=0$.
\end{proof}

\section{A schema, and two axes}\label{sec:schema}

\begin{definition}[Schema]\label{def:schema}
A \emph{parameterised verification problem} is given by a class of
\emph{perturbation parameters} $\mathcal P$, a quantifier prefix
$Q_1Q_2\in\{\exists,\forall\}^2$, and a condition $\Phi$ expressible by
linear constraints on the values of a network at a point. An instance
consists of an $F$-network $N$, a description of $\mathcal P$, an LP
specification $A$ for the inputs, and the data of $\Phi$; the question is
whether
\[
  Q_1\,p\in\mathcal P\quad Q_2\,x\in A:\qquad \Phi\bigl(N_p,x\bigr).
\]
\end{definition}

Three kinds of $\mathcal P$ occur. It may be \emph{trivial}, the network
being fixed; it may be an \emph{input-space} set, so that
$N_p(x)=N(x\oplus p)$ for some rational affine combination rule
$\oplus$; or it may be a \emph{parameter-space} set, so that $N_p$ is $N$
with different parameters. Table~\ref{tab:schema} on page~\pageref{tab:schema}
places every problem of this paper in the resulting grid.

The two axes are these.

\paragraph{Alternation.} The quantifier prefix determines the level of the
polynomial hierarchy, and Lemma~\ref{lem:uniform} below makes the
membership direction automatic for piecewise linear networks. All
previously studied problems have $\mathcal P$ trivial and a single
quantifier, hence live at the first level; $\MIN$ is the exception, and it
is at the second.

\paragraph{Arithmetic.} Whether $p$ enters the computation linearly or
multiplicatively determines whether the resulting constraint system is
linear or polynomial, and hence whether the problem lives in the
polynomial hierarchy or in the real hierarchy. An input-space parameter is
added to an input and enters exactly as an input does. A parameter-space
parameter multiplies an activation. Theorem~\ref{thm:wr} shows that this
distinction alone accounts for a jump from $\mathrm P$ to $\ER$-complete.

\begin{lemma}[Uniform membership]\label{lem:uniform}
Let $F$ be a finite set of semilinear activations with rational
coefficients. Then the graph
$\{(x,\bar y): \bar y \text{ is the vector of node values of } N \text{ on
} x\}$ of an $F$-network $N$ is defined by a quantifier-free formula of
real addition of size linear in the size of $N$. Consequently, if a
property of $N$ is expressible in the schema of
Definition~\ref{def:schema} with $k$ quantifier blocks and a trivial or
input-space parameter class, then deciding it lies in $\Sig k$ or
$\Pih k$ according to the leading quantifier.
\end{lemma}

\begin{proof}
Each activation $\sigma\in F$ has semilinear graph with rational
coefficients, so $y=\sigma(z)$ is equivalent to a Boolean combination
$\bigvee_{r}\bigl(y=a_rz+b_r\wedge \alpha_r\le z\le\beta_r\bigr)$ of
constant size, the number of pieces depending only on $\sigma$ and not on
$N$. The value $z$ entering a node is a rational linear form in the values
of the previous layer, so each node contributes a constant-size
quantifier-free formula, and $N$ contributes their conjunction, of size
linear in $N$.

For the second statement, write the property out. Writing $G$ for the
graph formula just constructed, a block $\exists x\in A$ becomes
\[
  \exists x\,\exists\bar y\ \bigl(G\wedge A(x)\wedge\cdots\bigr),
\]
and a block $\forall x\in A$ becomes
\[
  \forall x\,\forall\bar y\ \bigl(G\wedge A(x)\rightarrow\cdots\bigr);
\]
in either case the node variables join the block they belong to without
creating an alternation, because they are functionally determined by that
block's variables. An input-space parameter $p$ is a further variable of
its own block, and $N_p(x)=N(x\oplus p)$ is again linear. The result is a
$k$-alternation sentence of real addition of size linear in the instance,
and Theorem~\ref{thm:sontag} applies.
\end{proof}

\begin{remark}
Lemma~\ref{lem:uniform} subsumes the certificate arguments used throughout
the verification literature - ``guess the activation pattern and the
violated constraint, then solve a linear program'' - which are its
$k=1$ case. It is worth stating because it makes all membership results in
this paper immediate, and because it identifies exactly what is needed for
them: semilinearity of the activations, and a parameter that enters
linearly. Both hypotheses fail in Section~\ref{sec:weights}, and the
complexity changes accordingly.
\end{remark}

\section{Three properties that reduce to classified ones}\label{sec:notnew}

We begin with three properties that a practitioner would reach for first.
Each of them turns out to be linear-time interreducible with a problem
whose complexity is already known, so that none of them is new as a
\emph{computational} problem, however different they are as
specifications. We record this because the consequence is useful in both
directions: their complexity is classified by activation function already,
and a tool that decides equivalence or interval verification decides them
once a front-end has been written.

\subsection{Non-interference}

The oldest formal security property there is \cite{Goguen}, transplanted
from language-based security \cite{Sabelfeld}: the output must not depend
on the secret part of the input.

\begin{problem}[$\NI(F)$]\label{prob:NI}
Given an $F$-network $N$ on $n$ inputs and a set
$P\subseteq\{1,\dots,n\}$ of \emph{public} coordinates, the remaining
coordinates being \emph{secret}, decide whether
\[
  \forall x,x'\in\R^n:\quad x_P=x'_P\ \Longrightarrow\ N(x)=N(x').
\]
\end{problem}

\begin{proposition}\label{prop:NI}
$\NI(F)$ reduces to $\NE(F)$ in linear time for every $F$, so every upper
bound for $\NE(F)$ transfers to $\NI(F)$. Conversely $\NI(\ReLU)$ is
co-NP-hard already for $P=\emptyset$. Hence $\NI(\ReLU)$ is
co-NP-complete.
\end{proposition}

\begin{proof}
Let $S$ be the secret coordinates. On the enlarged input
$(x_P,x_S,x'_S)\in\R^{|P|}\times\R^{|S|}\times\R^{|S|}$ define
$N_1(x_P,x_S,x'_S):=N(x_P,x_S)$ and $N_2(x_P,x_S,x'_S):=N(x_P,x'_S)$.
Each is a copy of $N$ with some input edges rerouted and some inputs
ignored, hence an $F$-network of the same size; and $\NI(N,P)$ holds if
and only if $N_1\equiv N_2$. The construction is computable in linear
time.

For hardness take $P=\emptyset$, so that the premise $x_P=x'_P$ is vacuous
and non-interference says exactly that $N$ is a constant function. Given a
$3$-CNF formula $\varphi$, take $N:=\hat g_\varphi$ of
Definition~\ref{def:gadget}. By Corollary~\ref{cor:gadget} this network is
constant if $\varphi$ is unsatisfiable, and takes both the value $0$ and
the value $1$ if $\varphi$ is satisfiable. So $\NI(N,\emptyset)$ holds if
and only if $\varphi$ is unsatisfiable, and the construction is linear
time. Membership in co-NP for piecewise linear $F$ is
Lemma~\ref{lem:uniform} with $k=1$.
\end{proof}

\begin{corollary}[\cite{WurmJ}]\label{cor:NE}
$\NE(\ReLU)$ is co-NP-complete.
\end{corollary}

\begin{proof}
Membership is Lemma~\ref{lem:uniform} with $k=1$. For hardness compare
$\hat g_\varphi$ with the network computing the constant $0$: by
Corollary~\ref{cor:gadget} the two are equivalent if and only if $\varphi$
is unsatisfiable.
\end{proof}

The statement is not ours; network equivalence was classified in
\cite{WurmJ}, and Table~\ref{tab:schema} records it there. We give the
argument because it costs two lines once Corollary~\ref{cor:gadget} is
available, and because keeping it inside the paper is what lets
Proposition~\ref{prop:NI} be read without consulting \cite{WurmJ}.

\begin{remark}\label{rem:NIgeneral}
The hardness argument uses only an $F$-network that is constant precisely
when the source instance is negative. For any activation set $F$ admitting
such a gadget the same two lines transfer the lower bounds for
$\NNReach(F)$ to the complement of $\NI(F)$; in particular the complement
of $\NI(F)$ is then $\ER$-hard for every $F$ for which reachability is
$\ER$-hard \cite{Wurm23,Wurm24}, so that $\NI(F)$ is $\FER$-hard. We state
the $\ReLU$ case, which is the one the rest of the paper uses.
\end{remark}

\paragraph{Practical reading.} Two things follow. First, no new
\emph{solver} is required: a verifier that decides equivalence of two
$\ReLU$ networks decides non-interference once a front-end has duplicated
the network and tied the public inputs. That front-end is real work, and
the two specifications look nothing alike on the page; the reduction says
that the hard part is shared, not that the problems are the same. Second, and less comfortably, exact
non-interference is the wrong specification for a trained model. It is an
all-or-nothing condition, and a network trained on data in which the secret
correlates with anything at all will fail it, typically by a margin far too
small to matter. The specification practitioners actually want is metric
- the output should not depend \emph{appreciably} on the secret - and
that is not $\NE$; it is Problem~\ref{prob:IND} below, which sits a level
higher. The gap between the property that is easy to check and the property
one wants is, here, exactly one quantifier alternation. The alternative is
to leave the decision setting and measure leakage instead, as approximate
counting has been made to do for neural networks \cite{Baluta}; that moves the problem into a
counting class rather than up the hierarchy.

\subsection{Monotonicity and counterfactual fairness}

\begin{problem}[$\MONO(F)$, $\CF(F)$]\label{prob:mono}
$\MONO(F)$: given $N$, an LP specification $A$ and a coordinate $i$,
decide whether $N(x)\le N(x')$ componentwise for all $x,x'\in A$ with
$x'-x\in\R_{\ge0}e_i$. $\CF(F)$, after \cite{CFair}: given $N$, $A$ and a
rational affine
involution $\iota$ acting on a protected coordinate, decide whether
$\argmax_j N(x)_j=\argmax_j N(\iota x)_j$ for all $x\in A$ with
$\iota x\in A$.
\end{problem}

\begin{proposition}\label{prop:mono}
For semilinear $F$ with rational coefficients, $\MONO(F)$ and $\CF(F)$ are
in co-NP. Both are co-NP-hard for $F=\{\ReLU\}$, hence co-NP-complete.
\end{proposition}

\begin{proof}
\emph{Membership} is Lemma~\ref{lem:uniform} with $k=1$: both conditions
are single universal blocks over two copies of the network whose inputs are
tied by linear equations. (Note that Lemma~\ref{lem:uniform} is
insensitive to whether the conclusion uses strict or non-strict
inequalities, since the language of real addition contains both. This
matters because monotonicity is inherently non-strict, whereas $\VIP$ is
stated in \cite{WurmJ} with an open output polyhedron.)

\emph{Hardness of $\MONO(\ReLU)$.} Let $\varphi$ be a $3$-CNF formula and
$\hat g:=\hat g_\varphi$ its violation gadget, with values in $[0,1]$.
Build the $\ReLU$-network $M'$ on inputs $(z,t)\in\R^n\times\R$ with the
single output
\[
  M'(z,t)\ :=\ \ReLU\bigl(\hat g(z)-t\bigr),
\]
which is a legal $\ReLU$-network because $\hat g(z)$ is the value of a node
and $t$ an input, so the displayed node applies $\ReLU$ to a linear form in
earlier values. Take $A':=\{0\le t\le1\}$ and the coordinate $t$.

Fix $z$ and write $c:=\hat g(z)\in[0,1]$. If $c=0$ then
$M'(z,t)=\ReLU(-t)=0$ for every $t\in[0,1]$, so $M'(z,\cdot)$ is constant
and in particular non-decreasing. If $c>0$ then $M'(z,0)=c>0$ while
$M'(z,c)=\ReLU(0)=0$, and $(z,c)-(z,0)\in\R_{\ge0}e_t$ with $(z,c)\in A'$,
so monotonicity in $t$ fails. Hence $M'$ is monotone in $t$ on $A'$ if and
only if $\hat g\equiv0$, which by Corollary~\ref{cor:gadget} holds if and
only if $\varphi$ is unsatisfiable.

\emph{Hardness of $\CF(\ReLU)$.} Take the same $\hat g$, the input
specification $A':=\{0\le t\le1\}$, which is invariant under the rational
affine involution $\iota:(z,t)\mapsto(z,1-t)$, and the $\ReLU$-network
$M''$ with the two outputs
\[
  o_1(z,t):=\ReLU(1)=1,\qquad
  o_2(z,t):=\ReLU\bigl(2-2u(z,t)\bigr),\quad
  u(z,t):=\ReLU\bigl(\hat g(z)+t-1\bigr),
\]
where $o_1$ is a node with no incoming edges and bias $1$. For
$\hat g(z),t\in[0,1]$ we have $u\in[0,1]$ and hence $o_2=2-2u\in[0,2]$.

If $\varphi$ is unsatisfiable then $\hat g\equiv0$, so
$u(z,t)=\ReLU(t-1)=0$ on $A'$ and $o_2\equiv2>1=o_1$; the classification is
$2$ at every point of $A'$ and counterfactual fairness holds. If $\varphi$
is satisfiable, take $z^*$ with $\hat g(z^*)=1$, which exists by
Corollary~\ref{cor:gadget}. At $(z^*,1)$ we get $u=\ReLU(1)=1$ and
$o_2=0<1=o_1$, so the classification is $1$; at $\iota(z^*,1)=(z^*,0)$ we
get $u=\ReLU(0)=0$ and $o_2=2>1=o_1$, so the classification is $2$, and
fairness fails. Both constructions are computable in linear time.
\end{proof}

\begin{remark}\label{rem:monogeneral}
Neither reduction uses anything about $\ReLU$ beyond the gadget of
Corollary~\ref{cor:gadget} - a non-negative $F$-network, bounded by $1$,
vanishing identically exactly when the source instance is negative - and
the ability to apply an activation to a linear form in earlier node values.
For any $F$ admitting such a gadget the same constructions reduce the
corresponding reachability problem to $\MONO(F)$ and $\CF(F)$; in
particular both are $\FER$-hard for every $F$ for which reachability is
$\ER$-hard \cite{Wurm23,Wurm24}, which includes every non-linear polynomial
and the usual sigmoidal activations.

Note also that the earlier version of this argument passed through a
Heaviside threshold and produced a network with a negative output. That is
not available over $F=\{\ReLU\}$, since every node of a $\ReLU$-network
applies $\ReLU$ and all outputs are therefore non-negative; the
constructions above make the \emph{argument} of the final $\ReLU$ decrease
instead of the value, which is what allows the Heaviside hypothesis to be
dropped.
\end{remark}

\begin{remark}[Relation to rule checking]\label{rem:rules}
Monotonicity has been studied in this setting before, and for $\ReLU$
networks the co-NP-completeness in Proposition~\ref{prop:mono} is not new.
Rules extracted from a network are classified in \cite{WurmRules} for four
rule languages - propositional rules, oblique rules, $M$-of-$N$ rules and
monotonicity rules - under three questions: whether a rule holds of the
network, whether a set of rules is consistent, and whether a set of rules
is exhaustive. Rule verification is co-NP-complete for $\ReLU$ networks in
each of the four languages, monotonicity included, and checking a
monotonicity rule is shown there to be linear-time reducible to checking an
oblique rule. The monotonicity rules of \cite{WurmRules} are more general
than Problem~\ref{prob:mono}: they require that $Ax\le Ay$ imply that the
output does not decrease, for a matrix $A$, so the order need not be
coordinatewise, and Problem~\ref{prob:mono} is the case in which $A$
selects a single coordinate.

Proposition~\ref{prop:mono} adds two things. The hardness reduction is
stated for $\ReLU$ alone and constructs the network explicitly from the
formula, rather than assuming a threshold activation in order to build a
Boolean violation indicator; and by Remark~\ref{rem:monogeneral} the same
two gadgets carry the classification off $\ReLU$, giving $\FER$-hardness
wherever reachability is $\ER$-hard. The
framings are complementary in a further way. In \cite{WurmRules} the rule
language is the parameter and the interesting phenomena are differences
between languages - consistency and exhaustiveness of monotonicity rules
are in $\mathrm P$ there, while for the other three languages they are
co-NP-hard - whereas here monotonicity is one property among several and
the parameter of interest is the quantifier prefix. That the $\ReLU$ cases
agree is a consistency check on both readings.
\end{remark}

\begin{remark}
The two gadgets make the reduction concrete: a property whose prefix is a
single universal block over inputs to a fixed network is decided by
appending a constant number of threshold nodes and asking the interval
question. That is a statement about cost and not about meaning.
Monotonicity, counterfactual fairness and interval verification are three
different specifications, written over different data and checked by
different front-ends; only their difficulty coincides. Everything genuinely
new in the remainder of this paper either alternates quantifiers or
perturbs the network itself.
\end{remark}

\paragraph{Practical reading.} Monotonicity constraints are imposed by
regulation in credit scoring, insurance pricing and parts of clinical
decision support: a risk score must not fall when a risk factor rises.
Proposition~\ref{prop:mono} says that checking monotonicity of an
arbitrary trained network costs exactly what robustness certification
costs, so the same solvers apply; and, by Remark~\ref{rem:monogeneral},
that for sigmoidal networks the
problem is $\FER$-hard, which is a precise explanation of why monotonicity
checkers in the literature \cite{CertMono} restrict themselves to piecewise
linear models or else abandon completeness. The design response -
architectures that are monotone by construction, or counterexample-guided
learning that repairs violations during training \cite{CEGMono} - is, in
this light, the reasonable one: it replaces a co-NP-complete check by a
syntactic invariant or by an incremental search. The same reading applies
to fairness: individual fairness in the sense of Dwork et al.\ \cite{FTA}
is a Lipschitz condition and therefore an instance of the global robustness
problem already classified in \cite{WurmJ}, which is why its verification
\cite{VerifFair} reuses robustness machinery.

\section{Alternating input-space properties}\label{sec:trig}

We now come to properties whose quantifier prefix alternates. These are
genuinely outside the previously classified family, and the first of them
is the central security question about a model of unknown provenance.

\subsection{Backdoors}

A backdoored network behaves normally on ordinary inputs but misclassifies
any input carrying a \emph{trigger}: in the standard construction
\cite{BadNets} a small patch of fixed pixel values written over a
fixed
region of the image, which forces a target class whatever the rest of the
image is. Deciding whether a network is backdoored is therefore an
$\exists\forall$ question: \emph{does there exist a patch such that all
inputs carrying it are classified as $j$?}

\begin{problem}[$\TRIG(F)$]\label{prob:trig}
Given an $F$-network $N$, a \emph{patch} $I\subseteq\{1,\dots,n\}$, a
finite \emph{alphabet} $G\subseteq\Q$, an LP specification $A$ and a
target coordinate $j$, decide whether
\[
  \exists\,\tau\in G^{I}\ \ \forall\,x\in A:\qquad
  \argmax_{j'}N\bigl(x[I\mapsto\tau]\bigr)_{j'}=j,
\]
where $x[I\mapsto\tau]$ denotes $x$ with its $I$-coordinates overwritten
by $\tau$.
\end{problem}

\begin{theorem}\label{thm:trig}
$\TRIG(\ReLU)$ is $\Sig2$-complete.
\end{theorem}

\begin{proof}
\emph{Membership.} A trigger is an element of the finite grid $G^I$ and has
polynomial bit size, so it may be guessed. What remains is an instance of
$\VIP(\ReLU)$, which is in co-NP by Lemma~\ref{lem:uniform}. Hence
$\TRIG(\ReLU)\in\mathrm{NP}^{\mathrm{co\text-NP}}=\Sig2$.

\emph{Hardness.} We reduce from
$\exists\bar u\,\forall\bar v\,\psi(\bar u,\bar v)$ with $\psi$ in $3$-DNF,
which is $\Sig2$-complete. Let $\varphi:=\neg\psi$, a $3$-CNF over the same
variables with $m$ clauses, so that
\[
  \exists\bar u\forall\bar v\,\psi(\bar u,\bar v)
  \iff \exists\bar u:\ \varphi(\bar u,\cdot)\ \text{is unsatisfiable}.
  \tag{1}
\]

Let $M:=M_\varphi$ be the $\ReLU$-network of Lemma~\ref{lem:cnf}, which
we use through Corollary~\ref{cor:cnfmax}.

Now build the $\TRIG$ instance. Let $N$ be $M$ with two output
nodes, computing $c_1:=M$ and $c_2:=m-\frac12$, the latter by a node with
no incoming edges and bias $m-\frac12$; take $j:=2$, the index of $c_2$,
so that, with ties broken towards the smaller index, the target class is
achieved exactly when $M<m-\frac12$. Let the
patch $I$ consist of the input coordinates $z_a$ belonging to the
$\bar u$-variables, let $G:=\{0,1\}$, and let $A:=\R^n$.

If $\bar u$ is a Boolean assignment making $\varphi(\bar u,\cdot)$
unsatisfiable, take $\tau\in\{0,1\}^I$ to be $\bar u$. By
Corollary~\ref{cor:cnfmax}, applied with $S=I$ and $\sigma=\bar u$, we have
$M\le m-1<m-\frac12$ for every $x$, so the target class is achieved for
every $x$ and $\tau$ witnesses the $\TRIG$ instance.

Conversely let $\tau\in\{0,1\}^I$ witness the $\TRIG$ instance, and let
$\bar u$ be the Boolean assignment it encodes. If $\varphi(\bar u,\cdot)$
were satisfiable, then by Corollary~\ref{cor:cnfmax} there is an input $x$
carrying the patch $\tau$ with $M=m>m-\frac12$, so $c_1>c_2$ and the target
class fails at that $x$, contradicting the choice of $\tau$. So
$\varphi(\bar u,\cdot)$ is unsatisfiable, and by~(1) the $\Sig2$ instance
is positive.

The construction is computable in linear time, so $\TRIG(\ReLU)$ is
$\Sig2$-hard.
\end{proof}

\subsection{Why the quantisation matters}\label{sec:relax}

The restriction of the trigger to a finite alphabet is not a technical
convenience. It is exactly what makes the reduction sound, and the reason
is worth spelling out, because it has a practical counterpart.

Suppose the trigger were allowed to range over $\R^I$. Then the adversary
- the $\exists$ player - acquires a move with no Boolean counterpart.
Setting a patch coordinate $z_a:=\frac12$ makes both literal values of the
construction in the proof of Lemma~\ref{lem:cnf} vanish,
\[
  \alpha_a=\alpha_{\neg a}=0 :
\]
the literal $a$ and its negation are \emph{both} false. This is the content
of Lemma~\ref{lem:cnf}(iii), and it is precisely the failure of item (ii)
to be tight. Every clause
containing only $\bar u$-literals then has clause value $0$, whatever the
$\bar v$-part of the input does. In the game-theoretic reading, the
$\exists$ player is permitted to abstain on a variable rather than commit
to a truth value, and abstention is never worse for him than either
commitment, because clause values are monotone in the literal values.

\begin{proposition}\label{prop:relax}
The reduction of Theorem~\ref{thm:trig} is unsound for the relaxed problem:
there is a formula for which $\exists\bar u\forall\bar v\,\psi$ is false
while the relaxed $\TRIG$ instance the reduction produces is positive.
Moreover $\TRIG(\ReLU)$ with $G$ replaced by $\R$ is in $\Sig2$.
\end{proposition}

\begin{proof}
Take the $3$-CNF formula $\varphi:=(u\vee\neg u\vee u)$, with the single
$\bar u$-variable $u$, no $\bar v$-variables and $m=1$; equivalently
$\psi=\neg\varphi$, which is unsatisfiable. Since $\varphi(\bar u,\cdot)$
is satisfiable for both values of $\bar u$, the equivalence~(1) shows that
$\exists\bar u\forall\bar v\,\psi$ is false. The reduction of
Theorem~\ref{thm:trig} produces a $\TRIG$ instance whose patch is the
single coordinate $z_u$. The real trigger $\tau:=\frac12$ fixes the only
input coordinate at $\frac12$, so $M=0$ by Lemma~\ref{lem:cnf}(iii). As
$0<m-\frac12$, the target class is achieved at every input, and the
relaxed instance is positive.

Membership is Lemma~\ref{lem:uniform} with the prefix
$\exists\tau\,\forall x$.
\end{proof}

Whether the real-valued version is $\Sig2$-hard is open. We do not claim
that it is easier; what Proposition~\ref{prop:relax} establishes is that
hardness cannot be transferred along this reduction, so a proof would have
to defeat the abstention strategy rather than encode around it, and the two
problems are not interchangeable for the purpose of moving hardness between
them.

\paragraph{Practical reading.} Three consequences, of which the third is
the one we would press.

First, backdoor detection is provably harder than robustness
certification. A robustness verifier is a co-NP oracle. If $\TRIG$ were
decidable by polynomially many calls to such an oracle it would lie in
$\mathrm{P}^{\mathrm{NP}}=\Delta_2^{\mathrm P}$, and a $\Sig2$-complete
problem in $\Delta_2^{\mathrm P}$ collapses the polynomial hierarchy to
$\Delta_2^{\mathrm P}$. So no scheme that reduces backdoor scanning to a
bounded number of verification queries can be complete.

Second, this is consistent with the shape of the empirical literature.
Detection methods such as Neural Cleanse \cite{NeuralCleanse} are search
procedures over candidate triggers wrapped around a verification-like inner
test - that is, $\Sig2$ algorithms with the outer search done
heuristically. The theory says the outer search cannot be dispensed
with.

Third, and most concretely: gradient-based trigger reconstruction relaxes
the trigger to a continuous variable, and Proposition~\ref{prop:relax} says
that the relaxation is not a faithful stand-in for the quantised problem.
The relaxed search is permitted to occupy states - a patch value that
activates neither a feature nor its complement - that no realisable
trigger can occupy, and on the formula exhibited there those states alone
make a clean network look backdoored. This is a candidate explanation, at
the level of problem structure rather than optimisation, for the well-known
phenomenon that such methods reconstruct ``triggers'' that do not
correspond to any planted backdoor. A detection method whose search space
is the quantised patch space is solving the stated problem; one whose
search space is its convex relaxation is solving a different one, and the
same caution applies as in the finite-precision setting studied by Alsmann,
Lange and S\"alzer \cite{ALS}, where the arithmetic in which a network is
evaluated likewise changes the problem rather than approximating it.

\subsection{Universal vulnerability}

The dual prefix formalises the statement that a network is nowhere robust.
It is the certified form of the observation behind universal adversarial
perturbations \cite{UAP}, with the quantifiers the other way round: there
the same perturbation fools most inputs, here every input is fooled by some
perturbation.

\begin{problem}[$\UVUL(F)$]\label{prob:uvul}
Given $N$, an LP specification $A$, a rational $\varepsilon>0$ and a metric
$d\in\{d_1,d_\infty\}$, decide whether for every $x\in A$ there is $\tau$
with $\|\tau\|\le\varepsilon$ and
$\argmax_{j}N(x+\tau)_{j}\neq\argmax_{j}N(x)_{j}$.
\end{problem}

\begin{proposition}\label{prop:uvul}
$\UVUL(F)\in\Pih2$ for semilinear $F$ with rational coefficients.
\end{proposition}

\begin{proof}
Lemma~\ref{lem:uniform} with the prefix $\forall x\,\exists\tau$.
\end{proof}

Hardness is open. We note that the complement asks for the existence of a
single $\varepsilon$-robust input, which is the decision version of the
question an empirical robustness evaluation is really asking, and that the
same relaxation subtlety as in Section~\ref{sec:relax} may arise on the
inner quantifier.

\paragraph{Practical reading.} $\UVUL$ is the formal content of the claim
``this model is not robust anywhere in the operational domain'', which is
what an adversarial evaluation reports when every sampled input admits an
attack. The evaluation establishes the property on a finite sample; the
decision problem is the certified version, and it sits at the second level
of the hierarchy, one above the per-input robustness question
$\mathrm{CR}$ that verifiers answer. Certifying global non-robustness is
therefore harder than certifying local robustness - an asymmetry worth
knowing when a safety case tries to argue that a fallback mechanism is
always needed.

\section{Parameter-space uncertainty}\label{sec:weights}

Every robustness notion in the verification literature perturbs the
\emph{input}. Several of the failure modes that certification authorities
care about perturb the \emph{weights}: bit-flip and fault-injection
attacks on the memory holding the parameters \cite{BitFlip},
single-event upsets and other transient faults in the accelerators that
evaluate the network \cite{ErrorProp}, a concern in avionics and space
electronics; the device variation that
makes analog and in-memory computing fast but uncertain \cite{InMemory},
and the deviation introduced by post-training quantisation. Robustness to
weight perturbation has been formalised and studied empirically
\cite{WeightPerturb}; what follows is its worst-case complexity.

Two features of these fault models shape the right formalisation. A fault
is local, so each stored quantity should carry its own uncertainty set. And
a fault corrupts a \emph{stored parameter}, not an edge of the computation
graph: a parameter held once in memory and read on many edges - as in
every convolutional layer, and in any implementation backed by a shared
weight table - is corrupted on all of them at once. We therefore separate
parameters from edges.

\begin{definition}[Parameterised architecture]\label{def:arch}
A \emph{parameterised architecture} $\mathcal A$ consists of a layered
graph with edge set $E$, a parameter index set $\{1,\dots,P\}$, a
\emph{parameter map} $\pi:E\to\{1,\dots,P\}$, rational biases, and an
assignment of activations to nodes. A parameter vector $\theta\in\R^P$
induces the network $N_\theta$ in which edge $e$ carries the weight
$\theta_{\pi(e)}$. The \emph{sharing degree} of $\mathcal A$ is
$\max_p|\pi^{-1}(p)|$; an architecture of sharing degree $1$ is an ordinary
network, in which parameters and edges coincide.
\end{definition}

\begin{definition}[Parameter specification]\label{def:wspec}
A \emph{box parameter specification} for $\mathcal A$ consists of a nominal
vector $\theta^0\in\Q^P$ and budgets $\delta\in\Q_{\ge0}^P$, and denotes
\[
  W\ :=\ \prod_{p=1}^{P}\bigl[\theta^0_p-\delta_p,\ \theta^0_p+\delta_p
  \bigr].
\]
A parameter with $\delta_p=0$ is \emph{exact}, and the \emph{width} of $W$
is $\max_p\delta_p$. More generally a \emph{parameter specification} is any
LP instance over the variables $\theta_1,\dots,\theta_P$.
\end{definition}

\begin{problem}[$\WRE(F)$, $\WRA(F)$]\label{prob:wr}
Given a parameterised $F$-architecture $\mathcal A$, a parameter
specification $W$ and LP specifications $A,B$:
\[
  \begin{aligned}
    \WRE(F):&\quad \exists \theta\in W\ \ \exists x\in A:\ \
      N_\theta(x)\in B,\\[2pt]
    \WRA(F):&\quad \forall \theta\in W\ \ \forall x\in A:\ \
      N_\theta(x)\in B.
  \end{aligned}
\]
\end{problem}

$\WRE$ is the fault-injection attacker's question - is there a fault
within the model that drives the network into the unsafe region $B$ -
and $\WRA$ is the certifier's.

\begin{theorem}\label{thm:wr}
$\WRE(\{\mathrm{id}\})$ is $\ER$-complete. Hardness holds already for
instances in which
\begin{enumerate}[label=(\alph*),leftmargin=*]
\item $W$ is a box parameter specification,
\item the sharing degree is at most $2$,
\item every budget satisfies $\delta_p\le P^{-c}$, for any constant $c>0$
fixed in advance, and
\item $B$ is a conjunction of linear equations on output node values.
\end{enumerate}
The same holds verbatim with $\ReLU$ in place of $\mathrm{id}$. Moreover
$\WRA(\{\ReLU\})$ is $\FER$-complete, already under (a)-(c) and with $B$ a
single open half-space.
\end{theorem}

\begin{proof}
\emph{Membership.} Introduce a real variable for every parameter and every
node value. Each node contributes
$y_v=\sum_u \theta_{\pi(uv)}y_u+b_v$, a polynomial
equation with rational coefficients, bilinear in the unknowns; $W$, $A$
and $B$ contribute linear inequalities. The instance is positive iff the
resulting system has a real solution, so $\WRE(F)\in\ER$ for every
semilinear $F$, the semilinear activations contributing Boolean
combinations of linear conditions. $\WRA(F)\in\FER$ by complementation.

\emph{Hardness of $\WRE$.} We reduce from $\RETRINV$
(Theorem~\ref{thm:retrinv}); the range parameter $\delta$ is chosen at the
end, when we verify clause (c). Let the instance have $k$ variables
$\xi_1,\dots,\xi_k$ with promised intervals $I(\xi_i)\subseteq[\frac12,2]$
and constraint list $C$. The architecture has a single input node $x_1$ and
the input specification is $A:=\{x_1=1\}$, so that $x_1$ carries the value
$1$. All nodes are identity nodes. Each parameter introduced below is given
the box named with it, and every other parameter is exact with nominal
value $1$.
\begin{itemize}[leftmargin=*]
\item \emph{Variables.} For each $\xi_i$ a parameter $p_i$ with box
$I(\xi_i)$, and a node $v_i$ with a single incoming edge from $x_1$
carrying $p_i$ and bias $0$, so that $v_i=\theta_{p_i}$ ranges exactly over
$I(\xi_i)$. The node $v_i$ \emph{is} the variable $\xi_i$, and $\pi$ is
injective on these edges.
\item \emph{Inversions.} For each constraint $\xi_a\cdot\xi_b=1$ a fresh
parameter $q$ with box $I(\xi_b)$, read on exactly two edges: an edge from
$v_a$ into a new node $m$, and an edge from $x_1$ into a new node $o$, both
with bias $0$. Then $m=\theta_q\,v_a$ and $o=\theta_q$. Add to $B$ the
linear equations $o=v_b$ and $m=1$. The first forces $\theta_q=\xi_b$,
whereupon the second says precisely $\xi_a\xi_b=1$. The node $o$ makes the
\emph{value of a parameter} available as the \emph{value of a node}, which
is what lets the output specification state an equality between two
parameters; this is the one place where sharing is used.
\item \emph{Sums.} For each constraint $\xi_a+\xi_b=\xi_c$ add to $B$ the
linear equation $v_a+v_b=v_c$. No new parameter is needed.
\end{itemize}
Pad with exact edges of nominal weight $1$ so that all named nodes lie in
the output layer. The architecture has $O(k+|C|)$ nodes and $P=O(k+|C|)$
parameters, and sharing degree $2$, attained only by the parameters $q$.

For clause (c), note first that we may assume $|C|=O(k^3)$: there are at
most $k^3+k^2$ distinct constraints over $k$ variables, and duplicates may
be deleted without affecting satisfiability. Hence $P=O(k^3)$. Given the
constant $c$, apply Theorem~\ref{thm:retrinv} with range parameter
$\delta=O(k^{-(3c+1)})$; then every budget is at most $\delta\le P^{-c}$
for all sufficiently large $k$, and the finitely many smaller instances may
be decided outright. A pair $(\theta,x)$ witnessing $\WRE$ assigns to each
$v_i$ a value in $I(\xi_i)$ satisfying all the sum and inversion
constraints, that is, a satisfying assignment of the $\RETRINV$ instance;
and conversely. Hence $\WRE(\{\mathrm{id}\})$ is $\ER$-hard under (a)-(d).

\emph{The $\ReLU$ case.} Every node value occurring in the construction is
positive: $x_1=1$, $v_i\in I(\xi_i)\subseteq[\frac12,2]$,
$o=\theta_q\in[\frac12,2]$, $m=\theta_q v_a\in[\frac14,4]$, and the padding
edges have weight $1$; all biases are $0$. Declaring every node a $\ReLU$
node therefore changes no value, and the same instance witnesses
$\ER$-hardness of $\WRE(\{\ReLU\})$.

\emph{Hardness of $\WRA(\{\ReLU\})$.} Let $\ell_1,\dots,\ell_r$ be the
rational affine forms in the node values $v_i,o,m$ whose vanishing
expresses the conjunction $B$ above - one for each of the equations
$o-v_b$, $m-1$ and $v_a+v_b-v_c$ - so that the $\RETRINV$ instance is
satisfiable iff $\exists\theta\in W$ with $\ell_1=\dots=\ell_r=0$. Extend
the network by the nodes $\ReLU(\ell_s)$ and $\ReLU(-\ell_s)$ for $s\le r$
and the single output node
\[
  h\ :=\ \sum_{s=1}^{r}\bigl(\ReLU(\ell_s)+\ReLU(-\ell_s)\bigr),
\]
which is non-negative, so the $\ReLU$ on the output node is the identity
there. Then $h\ge0$ always, and $h(\theta)=0$ iff every $\ell_s$ vanishes.
All added edges are exact. Take $B':=\{h>0\}$, a single open half-space.
Since $A$ pins the input,
\begin{align*}
  \WRA(\mathcal A,W,A,B')\ \text{holds}
  \ &\iff\ \forall\theta\in W:\ h(\theta)>0\\
  \ &\iff\ \text{the $\RETRINV$ instance is unsatisfiable},
\end{align*}
and unsatisfiability of $\RETRINV$ is $\FER$-complete by
Theorem~\ref{thm:retrinv}. The reduction is computable in linear time and
preserves (a)-(c).
\end{proof}

\begin{remark}[What is not claimed]\label{rem:notclaimed}
Three restrictions of Theorem~\ref{thm:wr} are worth naming, because each
is a natural question we leave open.

\emph{Sharing.} Hardness uses sharing degree $2$: the parameter $q$ is read
both on the edge forming the product and on the edge exposing its value as
a node value. For an architecture of sharing degree $1$ under a box
specification there is no mechanism for equating two parameters, and we do
not know the complexity of $\WRE$; this is Open
Problem~\ref{op:sharing}. Under the general LP parameter specification of
Definition~\ref{def:wspec} the tying is expressible directly, so there the
theorem holds at sharing degree $1$ as well.

\emph{The identity case of $\WRA$.} The $\FER$-hardness argument aggregates
a conjunction of violations into a single half-space using $\ReLU$, which
an $\{\mathrm{id}\}$-network cannot do: an affine network cannot compute a
non-negative function vanishing exactly on a prescribed affine subspace.
Whether $\WRA(\{\mathrm{id}\})$ is $\FER$-hard is open.

\emph{Single-parameter faults.} If exactly one parameter is inexact, every
node value is affine in that parameter and the constraint system is
bilinear in one scalar unknown and the inputs. Theorem~\ref{thm:wr} says
nothing about this case, and we do not know its complexity.
\end{remark}

\begin{remark}[The point of the theorem]\label{rem:point}
For $\{\mathrm{id}\}$-networks, \emph{every} verification problem studied
in \cite{WurmJ} lies in $\mathrm P$: such a network computes an affine map
and each question is a linear program \cite[Prop.~3]{WurmJ}.
Theorem~\ref{thm:wr} says that making the parameters uncertain, within
boxes whose width shrinks inverse-polynomially in the size of the network,
takes the same networks from $\mathrm P$
to $\ER$-complete. Nothing about the activation function is involved,
which is precisely why the phenomenon is invisible in a classification
indexed by activation functions: the cause is that a parameter multiplies
an activation, so uncertainty in the parameters makes the network a
polynomial rather than a linear map of the unknowns. Tightening the
tolerance is therefore no help, which is clause (c) of the theorem and the
reason for reducing from the range-restricted $\RETRINV$ of \cite{AMS}
rather than from $\ETRINV$: the latter confines its variables to the fixed
box $[\frac12,2]$, whose relative width no rescaling of the network can
reduce.

The summary is worth stating as a slogan. \emph{The activation function
decides the complexity of input-space questions; the arithmetic decides the
complexity of parameter-space questions.} The first axis is classified; the
second is not, and on it even the trivial activation is $\ER$-complete.
\end{remark}

\begin{remark}
The result also sharpens the relation to learning. Training $\ReLU$
networks is $\ER$-complete \cite{AKM}, and the intuition usually offered is
that training searches over weights. Theorem~\ref{thm:wr} isolates that
intuition from everything else: it is neither the optimisation nor the
activation that produces $\ER$-hardness, but the product of a parameter and
an activation. Verification under parameter uncertainty is $\ER$-complete
for the simplest networks there are, and for arbitrarily tight tolerances.

It is worth contrasting this with the smoothed-analysis result of
Soltanalian \cite{Soltanalian}, which perturbs every parameter by
independent clipped Gaussian noise and shows that exact $\ReLU$
verification has no smoothed-polynomial complete algorithm unless
$\mathrm{NP}\subseteq\mathrm{BPP}$. There the parameters are randomised
once, as a preprocessing of the instance, and the problem asked afterwards
is still the input-space one; the content is that worst-case hardness
survives generic noise. Here the parameters are quantified inside the
problem, over an adversarial set, and the content is that the problem
leaves the polynomial hierarchy. Read together, the two say that parameter
uncertainty neither dissolves the difficulty of verification nor leaves it
where it was.
\end{remark}

\paragraph{Practical reading.} The consequence for tool building is
concrete and, we think, under-appreciated. Existing complete verifiers
encode a $\ReLU$ network as a mixed integer linear program \cite{Tjeng} or
solve it by SMT and branch and bound \cite{Marabou,BetaCrown}: the weights
are constants, the activations are the binary variables, and the whole
question is linear. It is tempting to certify against parameter faults by
adding the parameters as variables to the same encoding. Theorem~\ref{thm:wr}
says that this cannot work as stated, because the resulting constraints are
bilinear; and more than that, it says no reformulation can restore an MILP
encoding of polynomial size unless $\ER=\mathrm{NP}$. Certified parameter
robustness is not the input-space problem with a larger box. It is a
problem of a different arithmetic type, and the methods that apply to it
are those of polynomial optimisation - semidefinite relaxations,
branch-and-bound over parameter intervals, interval arithmetic with
refinement - all of which are incomplete in a way the MILP methods are
not.

For the fault models themselves the reading is as follows. Analog and
in-memory accelerators \cite{InMemory} give every parameter a small
non-zero budget simultaneously, which is exactly the setting of
Theorem~\ref{thm:wr}; and clause (c) says that tightening the device
tolerance does not help, since the problem is already $\ER$-complete at
inverse-polynomial width. A bit flip in a fixed-point weight
\cite{BitFlip} corrupts one stored parameter, which in a
convolutional layer is read on many edges at once; that the hardness
construction requires a parameter read on more than one edge is therefore
not an artefact but a feature of the fault model. The case in which a
single corrupted parameter is read on a single edge is a different
question, and Remark~\ref{rem:notclaimed} records it as open: there the
node values are affine in the unknown, and the reduction above says
nothing. What the theorem does establish is that the
simultaneous-tolerance model, which is the one analog hardware presents, is
outside the reach of the current complete-verification stack, and that is
consistent with the mitigations in use being redundancy and
error-correcting codes on the parameter memory rather than verification.

\section{Privacy and provenance}\label{sec:privacy}

We turn to three properties concerning what a model reveals and what it
can be made to forget. Each has a clear practical driver, and they occupy
three different levels of the picture.

\subsection{Inversion resistance}

Model inversion attacks \cite{ModelInversion} reconstruct an input, or a
sensitive attribute of it, from the model's output. Whether such reconstruction is possible at all
is a geometric question: does the output pin the input down?

\begin{problem}[$\INV(F)$]\label{prob:inv}
Given an $F$-network $N$, LP specifications $A$ and $B$, a rational
$\delta\ge0$ and a metric $d\in\{d_1,d_\infty\}$, decide whether
\[
  \diam\bigl(\{x\in A: N(x)\in B\}\bigr)\ \le\ \delta ,
\]
that is, whether all $x,x'\in A$ with $N(x),N(x')\in B$ satisfy
$d(x,x')\le\delta$.
\end{problem}

A network is \emph{inversion resistant} for the output region $B$ when the
answer is negative for small $\delta$: many admissible inputs produce
outputs in $B$, so observing an output in $B$ reveals little. The decision
problem as stated asks the opposite, for convenience of statement.

\begin{theorem}\label{thm:inv}
$\INV(\ReLU)$ is co-NP-complete.
\end{theorem}

\begin{proof}
\emph{Membership.} The condition is
$\forall x\,\forall x'\,(\cdots\rightarrow d(x,x')\le\delta)$, a single
universal block. For $d=d_\infty$ the conclusion
$\max_i|x_i-x'_i|\le\delta$ is the conjunction of the $2n$ linear
inequalities $x_i-x'_i\le\delta$ and $x'_i-x_i\le\delta$, so
Lemma~\ref{lem:uniform} with $k=1$ applies.

For $d=d_1$ the lemma does \emph{not} apply directly: the condition
$\sum_i|x_i-x'_i|\le\delta$ is equivalent to the conjunction of
$\sum_i s_i(x_i-x'_i)\le\delta$ over all sign vectors
$s\in\{-1,1\}^n$, and that quantifier-free formula has exponential size.
We argue instead on the complement, which asks for $x,x'\in A$ with
$N(x),N(x')\in B$ and $\sum_i|x_i-x'_i|>\delta$. Guess the activation
patterns of the two copies of $N$ and a sign vector
$s\in\{-1,1\}^n$, and verify the single linear inequality
$\sum_i s_i(x_i-x'_i)>\delta$ over the resulting polyhedron by linear
programming. This is sound because $\sum_i s_i t_i\le\sum_i|t_i|$ for every
$s$, and complete because $s_i=\operatorname{sign}(x_i-x'_i)$ attains the
sum; so the complement is in NP and $\INV(\ReLU)\in\mathrm{co\text-NP}$ for
$d_1$ as well. (Proposition~\ref{prop:invp} below gives a second route, the
case $p=1$ of a uniform argument for all $d_p$.)

\emph{Hardness.} We reduce the complement of $\NNReach(\ReLU)$, which is
co-NP-hard \cite{SL}. Let $(N,A,B)$ be a reachability instance with $N$ on
$n$ inputs. Let $N'$ be the network on $n+1$ inputs $(x,t)$ defined by
$N'(x,t):=N(x)$, obtained from $N$ by adding an input node with no outgoing
edges; let $A':=A\wedge 0\le t\le 1$ and $B':=B$; and let
$\delta:=\frac12$ with $d=d_\infty$. Then
\[
  \{(x,t)\in A': N'(x,t)\in B'\}\ =\ \{x\in A: N(x)\in B\}\times[0,1].
\]
If the reachability instance is positive this set contains $(x,0)$ and
$(x,1)$ for some $x$, so its $d_\infty$-diameter is at least $1>\delta$
and the $\INV$ instance is negative. If it is negative the set is empty
and the $\INV$ condition holds vacuously.\footnote{The reduction therefore
establishes hardness through vacuity: in the negative case the preimage is
empty and the diameter condition is never tested. Hardness under the
promise that the preimage is non-empty - which is the case of practical
interest, since an unreachable output region raises no privacy question -
does not follow from this argument, and we leave it open. The same caveat
applies to any reduction to a diameter or Hausdorff-distance problem that
passes through emptiness, including the natural first attempt at
Problem~\ref{prob:IND}.} The construction is linear time.
\end{proof}

\begin{proposition}\label{prop:invp}
For each fixed positive integer $p$, $\INV(\ReLU)$ is co-NP-complete when
$d$ is the metric $d_p$; in particular for the Euclidean metric $d_2$. The
same holds when $p$ is part of the input and encoded in unary.
\end{proposition}

\begin{proof}
Hardness is the reduction of Theorem~\ref{thm:inv}, which uses only that
the appended coordinate $t$ contributes a distance of $1$ between $(x,0)$
and $(x,1)$; that holds in every $d_p$.

For membership, consider the complement, which asks for $x,x'\in A$ with
$N(x),N(x')\in B$ and $d_p(x,x')>\delta$. Guess the activation patterns of
the two copies of $N$. What remains is a rational polyhedron
$P\subseteq\R^{2n}$ whose description has size polynomial in the instance,
together with the single constraint
\[
  f(x,x')\ :=\ \sum_{i=1}^{n}|x_i-x_i'|^{\,p}\ >\ \delta^{\,p},
\]
both sides of which are rational because $p$ is an integer. The hypothesis
that $p$ is fixed, or given in unary, is used here and only here: a $p$
encoded in binary would make $\delta^{\,p}$ and the summands rationals of
exponential bit length, and the verifier could not evaluate the certificate
in polynomial time. The function
$f$ is convex, being the $p$-th power of a norm composed with a linear map;
convexity needs $p\ge1$, which holds as $p$ is a positive integer.

We claim that the supremum of a convex function over a rational polyhedron
is either $+\infty$, witnessed by a recession direction, or attained at a
vertex of its pointed part. Write $P=Q+L$ with $L$ the lineality space and
$Q$ pointed. A convex function that is not constant on a line is unbounded
above on it, so either $f$ is constant along $L$, in which case the
supremum over $P$ equals the supremum over $Q$, or the supremum is
$+\infty$. On $Q=\operatorname{conv}(V)+\operatorname{cone}(R)$, a convex
function bounded above is non-increasing along every recession direction,
because a convex function of one variable that is bounded above on
$[0,\infty)$ is non-increasing there; so the supremum over $Q$ equals the
supremum over $\operatorname{conv}(V)$, which a convex function attains at
an extreme point, that is, at a vertex.

Vertices of a rational polyhedron and generators of its recession cone have
bit size polynomial in its description, so guessing one of them alongside
the activation patterns yields a polynomial-size certificate, and the
complement is in NP.
\end{proof}

\begin{remark}
It is worth saying why this does \emph{not} leave the polynomial hierarchy,
since the constraint $d_2(x,x')>\delta$ is semi-algebraic and not
semilinear, and elsewhere in this paper - Section~\ref{sec:weights} -
leaving linear arithmetic costs exactly that. The difference is where the
non-linearity sits. In Theorem~\ref{thm:wr} the products are \emph{inside}
the network, one per edge, and they compose: a chain of $d$ nodes realises
a monomial of degree $d$, which is what allows a network to encode an
arbitrary polynomial system. Here there is a single convex polynomial
inequality imposed on variables that are otherwise constrained only
linearly, and nothing composes. One convex constraint over a polyhedron is
not an existential theory of the reals.

Note also that the certificate is not an appeal to the trust region
subproblem. Maximising a convex quadratic over a polytope is NP-hard; what
gives the polynomial certificate is the \emph{location} of the maximiser,
not the cost of finding it.

For rational $p$ that is not an integer the argument breaks at the last
step: $\delta^{\,p}$ and the summands are then algebraic irrationals, and
comparing them is a question about sums of radicals, which is not known to
be decidable in polynomial time. Whether $\INV$ is in co-NP for such $p$ is
open.
\end{remark}

\paragraph{Practical reading.} $\INV$ is the certified form of the question
asked before publishing an embedding, a logit vector or a confidence
score: does releasing this number identify its source? Because the problem
is co-NP-complete, it is decidable by exactly the tools already used for
robustness, with a front-end that runs two copies of the network and
constrains their outputs to the same region. That is an unusually
favourable situation: a privacy property that costs no more than a
robustness property. The caveat is that $\INV$ is a worst-case, geometric
notion of leakage and is not a substitute for a statistical one such as
differential privacy \cite{DPBook}; a network can be inversion resistant in
this sense and still leak on the input distribution that actually occurs.

\subsection{Output indistinguishability}

Proposition~\ref{prop:NI} observed that exact non-interference is too
brittle to be a usable specification. The metric version compares what the
model's outputs look like on two populations.

\begin{problem}[$\IND(F)$]\label{prob:IND}
Given $N$, LP specifications $A_1,A_2$, a rational $\delta\ge0$ and
$d\in\{d_1,d_\infty\}$, decide whether the Hausdorff distance between
$N(A_1)$ and $N(A_2)$ is at most $\delta$; equivalently, whether
\[
  \forall x\in A_1\ \exists x'\in A_2:\ d(N(x),N(x'))\le\delta
  \qquad\text{and symmetrically.}
\]
\end{problem}

\begin{proposition}\label{prop:IND}
$\IND(F)\in\Pih2$ for semilinear $F$ with rational coefficients.
\end{proposition}

\begin{proof}
Lemma~\ref{lem:uniform} with the prefix $\forall x\,\exists x'$; the two
symmetric conjuncts share the prefix shape and may be combined.
\end{proof}

We conjecture $\Pih2$-completeness. The natural route is the reduction of
Theorem~\ref{thm:trig} with the roles of the quantifiers exchanged, and the
obstruction is the same relaxation phenomenon: the inner $\exists$ ranges
over a continuum.

\paragraph{Practical reading.} $\IND$ is the certified counterpart of an
audit for attribute inference or for membership inference
\cite{MembershipInference}. If the images of two
populations - inputs with and without a protected attribute, or members
and non-members of the training set - are within $\delta$ in Hausdorff
distance, then no adversary observing only the output can separate them by
more than $\delta$, whatever its computational power. That is a strong
guarantee of exactly the kind regulation asks for and empirical audits
cannot give. The price is one quantifier alternation: unlike
$\INV$, this is not a problem the current verification stack can be
pointed at. We regard closing the gap between $\NI$ (co-NP-complete but
unsatisfiable in practice) and $\IND$ ($\Pih2$ but meaningful) as the most
useful open problem in this section.

\subsection{Watermark removability}

Model watermarking embeds a signature either in the parameters or, in the
variant we formalise, in a set of key inputs with prescribed outputs
obtained by deliberate backdooring \cite{AdiWatermark}, so that the owner
of a stolen copy can demonstrate provenance. The security
question is whether an attacker can produce a model that is as good as the
original on real data but no longer answers the keys correctly. Because
watermark keys are deliberately out of distribution, the key inputs lie
outside the operational region.

\begin{problem}[$\WM(F)$]\label{prob:wm}
Given an $F$-network $N$, an LP specification $A$, a finite key set
$K=\{(k_1,y_1),\dots,(k_r,y_r)\}$ with $k_i\notin A$ and $N(k_i)=y_i$, and
a size bound $s\in\N$, decide whether there is an $F$-network $N'$ with at
most $s$ nodes such that $N'(x)=N(x)$ for all $x\in A$ and $N'(k_i)\neq
y_i$ for some $i$.
\end{problem}

\begin{proposition}\label{prop:wm}
If the weights of the candidate network $N'$ are restricted to rationals of
bit size polynomial in the instance, then $\WM(F)\in\Sig2$ for semilinear
$F$ with rational coefficients.
\end{proposition}

\begin{proof}
Guess $N'$, which is then of polynomial size, and verify
$N'|_A\equiv N|_A$ with a co-NP oracle by Lemma~\ref{lem:uniform}; the key
conditions are evaluated directly.
\end{proof}

The bit-size restriction is the same one implicit in the $\Pih2$ upper
bound for $\MIN$ in \cite[Thm.~11]{WurmJ}, and removing it is open in both
cases. We conjecture $\Sig2$-completeness, by adapting the
$\Sigma_2^{\mathrm P}$-hardness proof for minimum equivalent DNF
\cite{Umans} through the reduction chain relating $\NE$ and $\MIN$ in
\cite{WurmJ}.

\paragraph{Practical reading.} A watermarking scheme is usually validated
empirically, by showing that fine-tuning, pruning and extraction attacks
fail to remove the mark \cite{SoKWatermark}. $\WM$ is the corresponding worst-case statement, and its shape
tells us what such a validation can and cannot establish: since the
attacker's question is $\exists N'\forall x$, an equivalence checker alone
- which answers the inner $\forall$ - can never certify that a
watermark is irremovable, only that a \emph{specific} attacker's candidate
failed. Provenance claims built on watermarking therefore rest on a
$\Sig2$ assumption, and should be stated as such. The corresponding
question for the parameter-space attacker, in which $N'$ ranges over small
perturbations of $N$ rather than over all small networks, inherits the
arithmetic of Section~\ref{sec:weights} and is the more realistic threat
model.

\section{Repair}\label{sec:repair}

A network that fails certification is not discarded; it is patched, and
the patch must be small enough that the rest of the safety case survives.

\begin{problem}[$\REP(F)$]\label{prob:rep}
Given a parameterised $F$-architecture $\mathcal A$ with nominal parameters
$\theta^0$, LP specifications $A,B$, a budget $k\in\N$ and a parameter
specification $W$, decide whether there is $\theta\in W$ differing from
$\theta^0$ in at most $k$ coordinates with $N_\theta(x)\in B$ for all
$x\in A$.
\end{problem}

\begin{proposition}\label{prop:rep}
Let $F$ be semilinear with rational coefficients. If $W$ confines the
modified parameters to a finite rational grid, then $\REP(F)$ lies in
$\Sig2$. If $W$ is an arbitrary LP parameter specification, then $\REP(F)$
lies in $\EFR$.
\end{proposition}

\begin{proof}
In the first case guess the $k$ modified parameters and their grid values,
then apply Lemma~\ref{lem:uniform} with $k=1$ to the remaining $\VIP$
instance. In the second, the condition is
$\exists\theta\,\forall x\,\forall\bar y$ over a system of polynomial
constraints, which is the defining form of $\EFR$ \cite{BeyondER}.
\end{proof}

We conjecture completeness in both cases.

\paragraph{Practical reading.} Repair tools
\cite{ProvableRepair,MinimalMod} search over weights continuously, by
linear programming or gradient descent with a verification oracle in the
loop, and then round. Proposition~\ref{prop:rep} says that the rounding is not a
detail of implementation but a change of problem: with a quantised
parameter search the problem is in the polynomial hierarchy, and with a
continuous one it is in the real hierarchy one level above $\ER$. This is
the same phenomenon as in Section~\ref{sec:relax}, and as in the
finite-precision setting of \cite{ALS}, in the parameter space rather than
the input space, and it suggests the same methodological conclusion -
that a tool should decide, and state, which of the two problems it is
solving.

\section{Summary and open problems}

\begin{table}[h]
\centering
\begin{tabular}{llll}
\toprule
problem & prefix & parameter & status\\
\midrule
$\NNReach$ & $\exists x$ & - & NP-complete \cite{SL}\\
$\VIP$, $\NE$ & $\forall x$ & - & co-NP-complete \cite{WurmJ}\\
$\NI$ & $\forall x$ & - & co-NP-complete (Prop.~\ref{prop:NI})\\
$\MONO$, $\CF$ & $\forall x$ & - & co-NP-complete
(Prop.~\ref{prop:mono}, \cite{WurmRules})\\
$\INV$ & $\forall x\forall x'$ & - & co-NP-complete
(Thm.~\ref{thm:inv})\\
$\MIN$ & $\exists N'\forall x$ & structure & in $\Pih2$ \cite{WurmJ}\\
$\TRIG$ (quantised) & $\exists\tau\forall x$ & input &
$\Sig2$-complete (Thm.~\ref{thm:trig})\\
$\TRIG$ (real) & $\exists\tau\forall x$ & input & in $\Sig2$; hardness
open\\
$\UVUL$ & $\forall x\exists\tau$ & input & in $\Pih2$\\
$\IND$ & $\forall x\exists x'$ & - & in $\Pih2$\\
$\WM$ & $\exists N'\forall x$ & structure & in $\Sig2$
(Prop.~\ref{prop:wm})\\
$\REP$ (quantised) & $\exists\theta\forall x$ & parameters & in $\Sig2$\\
$\REP$ (continuous) & $\exists\theta\forall x$ & parameters & in $\EFR$\\
$\WRE$ & $\exists\theta\exists x$ & parameters & $\ER$-complete
(Thm.~\ref{thm:wr})\\
$\WRA$ ($\ReLU$) & $\forall\theta\forall x$ & parameters &
$\FER$-complete (Thm.~\ref{thm:wr})\\
$\WRA$ ($\mathrm{id}$) & $\forall\theta\forall x$ & parameters &
in $\FER$; hardness open\\
\bottomrule
\end{tabular}
\caption{The problems of this paper in the schema of
Section~\ref{sec:schema}. All completeness entries are for $\ReLU$ unless
stated otherwise. Rows above the $\MIN$ line are first-level;
rows with a parameter-space parameter leave the polynomial hierarchy.}
\label{tab:schema}
\end{table}

The open problems we would rank highest are the following.

\begin{enumerate}[label=(\arabic*),leftmargin=*]
\item\label{op:sharing} Is $\WRE$ still $\ER$-hard for architectures of
sharing degree $1$ - that is, when every parameter is read on a single
edge, so that a box specification gives each of them an independent
tolerance? The construction of Theorem~\ref{thm:wr} needs one parameter
read twice, and without it the node values available are sums of path
monomials in which each deeper parameter occurs in a single position, which
appears to express interval constraints rather than products. Either a
hardness proof or a polynomial-hierarchy upper bound would be interesting.
\item What is the complexity of the single-corrupted-parameter fault model,
in which exactly one budget is non-zero? The node values are then affine in
the unknown, so Theorem~\ref{thm:wr} does not apply; NP-hardness looks
likely and $\ER$-hardness unlikely, but we have neither. Relatedly, is
$\WRA(\{\mathrm{id}\})$ $\FER$-hard (Remark~\ref{rem:notclaimed})?
\item Is real-valued $\TRIG$ $\Sig2$-hard? By Proposition~\ref{prop:relax}
a proof must defeat the abstention strategy, and a proof of the opposite -
that the relaxation is strictly easier - would be more interesting still,
since it would say that continuous trigger search is solving a tractable
shadow of the stated problem.
\item Is $\IND$ $\Pih2$-complete, and is there a specification between
$\NI$ and $\IND$ that is both meaningful for trained models and decidable
at the first level? Is $\INV$ co-NP-hard under the promise that the
preimage is non-empty?
\item Are $\WM$ and $\REP$ complete for their classes, and can the
bit-size restriction in Proposition~\ref{prop:wm} - and in the
corresponding bound for $\MIN$ - be removed?
\item The counting versions. Quantitative information flow, demographic
parity and probabilistic certification all ask for a proportion rather than
a verdict, and none of the problems above has been studied in its counting
form, although approximate counting has been applied to neural networks
with security applications in mind \cite{Baluta}. We expect $\#\mathrm P$-completeness for the first-level problems
over $\ReLU$ networks, and the sigmoidal case appears to need a counting
class over $\ER$ that does not yet exist.
\end{enumerate}

\section*{Declarations}

\paragraph{Competing interests.} The author declares that he has no
competing interests.

\paragraph{Use of generative artificial intelligence.} A large language
model (Claude, Anthropic) was
used as a research assistant during the preparation of this manuscript.
Its use went beyond language editing: it was used to propose candidate
formalisations of several of the decision problems studied here, to draft
arguments and expository text, and to assemble candidate references. Every
definition, problem statement, theorem and proof in this paper has been
checked by the author, who takes full responsibility for the correctness of
the mathematical content and for the manuscript as a whole. The model is
not an author and is not accountable for the work.



\end{document}